\documentclass[authoryear,12pt,3p]{jowarticle}
\usepackage{graphicx}
\usepackage{amsmath}
\usepackage{amssymb}
\usepackage{amsfonts}
\usepackage{amsthm}
\usepackage{natbib}
\usepackage{setspace}
\usepackage[all]{xy}
\usepackage{enumitem}
\usepackage{titlesec}
\usepackage{mathrsfs}
\makeatletter
\renewcommand\section{\@startsection{section}{1}{\z@}{-3.25ex plus -1ex minus -.2ex}{1.5ex plus .2ex}{\normalsize\bf}}
\renewcommand\subsection{\@startsection{subsection}{2}{\z@}{-3.25ex plus -1ex minus -.2ex}{1.5ex plus .2ex}{\normalsize\bf}}
\renewcommand\subsubsection{\@startsection{subsubsection}{3}{\z@}{-3.25ex plus -1ex minus -.2ex}{1.5ex plus .2ex}{\normalsize\bf}}
\makeatother

\numberwithin{equation}{section}

\newtheorem{thm}{Theorem}

\newtheorem{prop}[thm]{Proposition}

\newtheorem{fact}[thm]{Fact}
\newtheorem{conj}[thm]{Conjecture}

\begin{document}
\begin{frontmatter}
\title{Would Two Dimensions be World Enough for Spacetime?}
\author{Samuel C. Fletcher} \ead{scfletch@umn.edu}
\address{Department of Philosophy \\ University of Minnesota, Twin Cities}
\author{JB Manchak}\ead{jmanchak@uci.edu}
\author{Mike D. Schneider}\ead{mdschnei@uci.edu}
\author{James Owen Weatherall}\ead{weatherj@uci.edu}
\address{Department of Logic and Philosophy of Science\\ University of California, Irvine}

\begin{abstract}We consider various curious features of general relativity, and relativistic field theory, in two spacetime dimensions.  In particular, we discuss: the vanishing of the Einstein tensor; the failure of an initial-value formulation for vacuum spacetimes; the status of singularity theorems; the non-existence of a Newtonian limit; the status of the cosmological constant; and the character of matter fields, including perfect fluids and electromagnetic fields.  We conclude with a discussion of what constrains our understanding of physics in different dimensions.\end{abstract}

\end{frontmatter}

\section{Introduction}

Philosophers of physics---and conceptually-oriented mathematical physicists---have gained considerable insight into the foundations and interpretation of our best physical theories, including general relativity, non-relativistic quantum theory, and quantum field theory, by studying the relationships between these theories and other ``nearby'' theories.  For instance, one can better understand general relativity by studying its relationship to Newtonian gravitation, particularly in the form of geometrized Newtonian gravitation (i.e. Newton-Cartan theory);\footnote{For background on geometrized Newtonian gravitation, see \citet{Trautman} and (especially) \citet[Ch.~4]{MalamentGR}.  For projects that aim to use this theory to provide new insight into general relativty, see, for instance, \citet{Cartan1,Cartan2}, \citet{Friedrichs}, \citet{Friedman}, \citet{WeatherallSingularity, WeatherallSGP,WeatherallPuzzleball,WeatherallConservation}, \citet{Manchak+Weatherall}, \citet{Dewar+Weatherall}, and \citet{EhlersLimit2}.} or by considering its relationship to other relativistic theories of gravitation.\footnote{See, for instance, \citet{BrownPR}, \citet{KnoxFF,KnoxEmergence}, \citet{Pitts}, or \citet{WeatherallConservation}.}  Likewise, formulating classical mechanics in the language of Poisson manifolds provides important resources for understanding the structure of Hilbert space and quantum theory.\footnote{See, for instance, \citet{Weyl} and \citet{LandsmanMT,LandsmanBohr} for mathematical treatments of the main issues; for examples of how these ideas have been applied by philosophers, see, for instance, \citet{FeintzeigAlgebra} and \citet{FLRW}.}  And thinking about classical field theory using nets of *-algebras on spacetime can help us better understand quantum field theory.\footnote{See, for instance, \citet{Fredenhagen+Brunetti}, \citet{Rejzner}, and \citet{FeintzeigUI, FeintzeigPOs}.}

The key feature of projects of the sort just described is that they are comparative: one draws out features of one theory by considering the ways in which it differs from other theories.  But there is a closely allied project---or better, strategy for conceiving of projects---that, though often taken up by mathematical physicists, has received considerably less attention in the philosophy of physics literature.\footnote{To our knowledge, the projects that come closest to this strategy are those that evaluate arguments that spacetime must have a certain dimensionality \citep{CallenderDims}; or those that consider the details of constructive quantum field theory, which often considers lower-dimensional models \citep{RuetscheIQT, Hancox-Li}.}  This strategy is to study the foundations of a physical theory by considering features of that same theory in other dimensions.  Doing so can provide insight into questions concerning, for instance, whether inferences about the structure of the world that make use of the theory in fact follow from the theory itself, or if they depend on ancillary assumptions.  For instance, (vacuum) general relativity in four dimensions is, in a certain precise sense, deterministic.  But as we argue in what follows, this feature depends on dimensionality; in two dimensions the theory, at least on one understanding, does not have a well-posed initial value formulation.

A detailed study of the physics of different dimensions can also reveal striking disanalogies between physics in different dimensions, which can then inform other projects.  For instance, it is common in the mathematical physics literature to consider quantizing field theories---including general relativity---in lower dimensions.\footnote{See, for instance, \citet{Glimm+Jaffe}; for a discussion of quantum gravity in particular, see \citet{Carlip}.}  Doing so can provide important hints at what a full theory of quantum gravity, say, might look like.  Moreover, there is a temptation to try to draw preliminary philosophical morals about our own universe from these quantum theories in lower dimensions---particularly among philosophers who prefer to work with mathematically rigorous formulations of theories, which in the case of quantum field theories are only available in lower dimensions.  But there are also reasons to be cautious about such hints: if classical theories, including general relativity, have very different features in different dimensions, the inferences we can draw about their quantum counterparts in those other dimensions may not carry over to the four dimensional case.

In what follows, we investigate the features of general relativity in two spacetime dimensions, on several ways of understanding what that might mean.  In the first instance, we suppose that Einstein's equation holds in all dimensions.  As we will show, the resulting theory is strikingly different, in a number of important ways, from the standard four dimensional theory.  Of course, that theories can differ dramatically in different dimensions is hardly news---especially to the experts in mathematical physics who work on these theories in fewer (or more) than four dimensions---and it is well-known that general relativity in two dimensions is ``pathological'' or (arguably) ``trivial''.  But there are some features that we discuss below that, to our knowledge, have not been drawn out in detail in the literature---including, for instance, the status of the initial value formulation and the non-existence of a Newtonian limit (where Newtonian gravitation is generalized by assuming that the geometrized Poisson equation holds in all dimensions).  Moreover, in our view it is valuable to collect these features of the two-dimensional theory together in one place, and to reflect on what they can teach us about the structure of general relativity more generally.  They also raise the question of what it means to identify theories across dimensions, particularly when the ostensibly ``same'' theory can have very different qualitative features in different dimensions.

In the next section, we will discuss the status of the Einstein tensor---which vanishes identically in two dimensions---and Einstein's equation (without cosmological constant).  In a sense, this is the principal feature of two-dimensional general relativity from which the other strange features follow.  In the following section, we will discuss the status of the initial value formulation and singularity theorems in two dimensions.  Next we will consider Newtonian gravitation in two dimensions, generalized as noted above, and show that it is not the classical limit of general relativity.  In the following section, we will consider what happens when one includes a cosmological constant, exploring the consequences for the character of some matter fields in two dimensions.  We will then discuss what it means to generalize a theory to different dimensions, by considering various arguments about alternative formulations of the theory in two dimensions.  We conclude by arguing that the discussion here of how to generalize a theory to other dimensions raises questions for a common view according to which to interpret a physical theory is to characterize the space of possibilities allowed by that theory.

\section{Einstein's Tensor and Einstein's Equation}

We begin with a few preliminaries concerning the relevant background formalism of general relativity.\footnote{The reader is encouraged to consult \citet{Hawking+Ellis}, \citet{Wald}, and \citet{MalamentGR} for details.} An $n$-dimensional relativistic {\em spacetime} (for $n \geq 2$) is a pair $(M, g_{ab})$ where $M$ is a smooth, connected $n$-dimensional manifold and $g_{ab}$ is a smooth, non-degenerate, pseudo-Riemannian metric of Lorentz signature $(+,-,...,-)$ defined on $M$.\footnote{We also assume $M$ to be Hausdorff and paracompact.  All objects that are candidates to be smooth in what follows are assumed to be so, even when not mentioned explicitly.}

For each point $p \in M$, the metric assigns a cone structure to the tangent space $M_p$. Any tangent vector $\xi^a$ in $M_p$ will be {\em timelike} if $g_{ab}\xi^a\xi^b>0$, {\em null} if $g_{ab}\xi^a\xi^b=0$, or {\em spacelike} if $g_{ab}\xi^a\xi^b<0$. Null vectors delineate the cone structure; timelike vectors are inside the cone while spacelike vectors are outside. A {\em time orientable} spacetime is one that has a continuous timelike vector field on $M$.  A time orientable spacetime allows one to distinguish between the future and past lobes of the light cone. In what follows, it is assumed that spacetimes are time orientable and that an orientation has been chosen.

For some open (connected) interval $I\subseteq \mathbb{R}$, a smooth curve $\gamma: I\rightarrow M$ is {\em timelike} if the tangent vector $\xi^a$ at each point in $\gamma[I]$ is timelike. Similarly, a curve is {\em null} (respectively, {\em spacelike}) if its tangent vector at each point is null (respectively, spacelike). A curve is {\em causal} if its tangent vector at each point is either null or timelike. A causal curve is {\em future directed} if its tangent vector at each point falls in or on the future lobe of the light cone.  A curve $\gamma: I\rightarrow M$ in a spacetime $(M,g_{ab})$ is a {\em geodesic} if $\xi^a \nabla_a \xi^b=\textbf{0}$, where $\xi^a$ is the tangent vector to $\gamma$ and $\nabla_a$ is the unique derivative operator compatible with $g_{ab}$.

The fundamental dynamical principle of general relativity is known as \emph{Einstein's equation}. In four dimensions, Einstein's equation may be written, without cosmological constant, in natural units as
\begin{equation}\label{EinsteinEquation}
R_{ab}-\frac{1}{2}g_{ab}R =8\pi T_{ab}.
\end{equation}
Here $R_{ab} = R^n{}_{abn}$ is the Ricci tensor associated with $g_{ab}$ and $R=R^a{}_a$ is the curvature scalar.  The left-hand side of this equation is known as the \emph{Einstein tensor}, often written $G_{ab}$; the right-hand side is the sum of the energy-momentum tensors associated with all matter present in the universe and their interactions.

In the first instance, we generalize general relativity to other dimensions by taking this expression to relate curvature and energy-momentum in arbitrary dimensions.  (We will return to this proposal in sections \ref{sec:cosConst} and \ref{sec:alternatives} and consider other possibilities.)  In particular, define, in a spacetime of any dimension, the Einstein tensor to be $G_{ab}=R_{ab}-\frac{1}{2} g_{ab}R$.

We have the following immediate proposition.
\begin{prop}\label{EinsteinTensor}
Let $(M,g_{ab})$ be a two-dimensional spacetime.  Then $R_{ab}=\frac{1}{2}Rg_{ab}$ and $G_{ab}=\mathbf{0}$.
\end{prop}
\begin{proof}Given a pseudo-Riemannian manifold of any dimension $n\geq2$, the Riemann tensor $R_{abcd}=g_{an}R^{n}{}_{bcd}$ is antisymmetric in the first two indices and in the last two indices: $R_{abcd}=R_{[ab][cd]}$.  It follows that $R_{abcd}$ can be written as a linear combination of outer products of two-forms.  But the space of two-forms on a two-dimensional manifold is one-dimensional, and so we have $R_{abcd}=f\epsilon_{ab}\epsilon_{cd}$, where $\epsilon_{ab}$ is either volume element on $M$ determined by $g_{ab}$.\footnote{Compare with \citet[p.~54]{Wald}.}  (Observe that $f$ is independent of the choice, since it is the square of the volume element that appears here; if $M$ is non-orientable, one can proceed locally.)  It follows, from standard identities concerning volume elements, that $R_{ab} = f\epsilon^n{}_a\epsilon_{bn} = fg_{ab}$, and thus $R=2f$.  Thus $G_{ab}=R_{ab}-\frac{1}{2}Rg_{ab}=fg_{ab}-\frac{1}{2}(2f)g_{ab}=\mathbf{0}$. \end{proof}

Let us now draw out some of the consequences of Prop.~\ref{EinsteinTensor}.  First, it follows, at least on this generalization of general relativity to other dimensions, that \emph{every} two-dimensional spacetime $(M,g_{ab})$ is a \emph{vacuum} solution to Einstein's equation.  It is tempting to conclude from this result that there can be no matter in an two-dimensional spacetime; in fact, this holds only if we assume that the total energy-momentum tensor vanishes only if all matter fields vanish, which is not necessarily true in two dimensions.  (We return to this point in section \ref{sec:cosConst}.)

Here is another consequence.  Recall that in four dimensions, vacuum spacetimes always have vanishing Ricci curvature ($R_{ab}=
\mathbf{0}$), but in general they have non-vanishing \emph{Weyl curvature}, which is defined by
\begin{equation}\label{WeylTensor}
C_{abcd}=R_{abcd}-\frac{2}{n-2}\left(g_{a[d}R_{c]b}+g_{b[c}R_{d]a}\right)-\frac{2}{(n-1)(n-2)}Rg_{a[c}g_{d]b}.
\end{equation}
But the situation in two dimensions is strikingly different.  In particular, Prop.~\ref{EinsteinTensor} immediately implies that being a vacuum solution in two dimensions does \emph{not} imply that a spacetime has vanishing Ricci curvature.  To the contrary, all solutions with non-vanishing Ricci curvature are vacuum solutions.\footnote{One can confirm that there exist two-dimensional spacetimes with non-vanishing Ricci curvature by direct computation; an example is offered in section \ref{sec:Newton}.}

Conversely, there is a sense in which vacuum solutions in two dimensions have vanishing Weyl curvature.  To make this idea precise requires some work.  The reason is that the definition of the Weyl tensor in Eq.~\eqref{WeylTensor} only makes sense in 3 or more dimensions.  Still, we can ask whether there is any tensor in two dimensions that might be a suitable analog for the Weyl tensor.  For instance, the Weyl tensor is defined as the ``trace-free'' part of the Riemann tensor, i.e., the trace-free tensor one gets by subtracting appropriate terms linear in the Ricci tensor and the scalar curvature.  To ask if a candidate Weyl tensor is available in two dimensions, then, we might consider tensors of the form
\[C_{abcd}=R_{abcd}-k\left(g_{a[d}R_{c]b}+g_{b[c}R_{d]a}\right)-\ell Rg_{a[c}g_{d]b},
\]
where $k$ and $\ell$ are unspecified constants.  We then ask: for which values of $k$ and $\ell$ is this quantity trace-free, in the sense that $C^n{}_{abn}=\mathbf{0}$?

The first thing to observe is that since $R_{ab}=\frac{1}{2}Rg_{ab}$, the terms proportional to $k$ and $\ell$ are in fact proportional to one another.  Thus, by redefining constants, it is sufficient to consider only tensors of the form
\[C_{abcd}=R_{abcd}-k Rg_{a[c}g_{d]b}.
\]
Taking the trace of both sides, we find that this will be trace-free just in case
\[
\mathbf{0} = R_{ab} + \frac{k}{2}R g_{ab} = \frac{R(1+k)}{2}g_{ab},
\]
i.e., if $k=-1$.  A short calculation similar to the proof of Prop.~\ref{EinsteinTensor} then yields that
\[
C_{abcd}=R_{abcd}+ Rg_{a[c}g_{d]b} = \frac{1}{2}R\epsilon_{ab}\epsilon_{cd} -\frac{1}{2}R\epsilon_{ab}\epsilon_{cd} = \mathbf{0}.
\]

We thus have a sense in which the natural candidate for a two-dimensional Weyl tensor is identically zero.\footnote{Other arguments are available to support the claim that any viable candidate for the Weyl tensor must vanish identically.  For instance, in the proof of Prop.~\ref{EinsteinTensor}, we show that $R^{a}{}_{bcd}$ is proportional to the scalar curvature, and thus vanishes if and only if its trace does.}  To check that this makes sense, recall that the vanishing of the Weyl tensor in four or more dimensions is associated with a spacetime being ``conformally flat'', in the sense that, at least locally, $g_{ab}=\Omega^2\eta_{ab}$, where $\eta_{ab}$ is flat.\footnote{In three dimensions, the Weyl tensor is defined, but vanishes identically; conformal flatness is equivalent to the vanishing of a different tensor, known as the Cotton tensor: $C_{abc} = \nabla_c R_{ab}-\nabla_b R_{ac} - \frac{1}{2(n-1)}(\nabla_b R_{ac}-\nabla_c R_{ab})$.  Observe that unlike the Weyl tensor, the expression for the Cotton tensor is well-defined in two dimensions, but it follows from Prop. \ref{EinsteinTensor} that it vanishes identically.  We are grateful to Brian Pitts for pointing out an error related to the Weyl and Cotton tensors in an earlier draft.}  It turns out that something closely related holds in two dimensions: there is no non-vanishing candidate for a Weyl tensor, and all two-dimensional spacetimes are (locally) conformally flat:\footnote{The ``local'' qualification is very important here, and often dropped in informal discussions.  For instance, it is a consequence of the Gauss-Bonnet theorem that there is no flat metric on the 2-sphere.}
\begin{fact}Let $(M,g_{ab})$ be a two dimensional spacetime.  Then in some neighborhood of every point, there exists a flat metric $\eta_{ab}$ such that $g_{ab}=\Omega^2\eta_{ab}$ in that neighborhood.\end{fact}

In sum, all two-dimensional relativistic spacetimes are vacuum solutions to Einstein's equation in two dimensions.  Moreover, these spacetimes may be Ricci curved, but are always conformally flat---even in the presence of non-trivial Riemann curvature.  This is precisely the opposite of what we are accustomed to in four dimensions.

\section{Determinism and Singularities}\label{sec:detSing}

We now turn to another, in many ways more substantial, set of consequences of Prop.~\ref{EinsteinTensor}.  Recall that in four dimensions, general relativity admits a well-posed initial value formulation for vacuum solutions.\footnote{The non-vacuum case depends crucially on the type of matter being considered \citep[pp.~266--267]{Wald}.} In other words, if appropriate initial data are specified, then there exists a unique vacuum solution that is the maximal evolution of that data. The upshot of this result is a clear sense in which ``Laplacian determinism holds'' in general relativity \citep[p.~188]{EarmanPD}.

To make this precise, we require some further preliminaries.  We say two spacetimes $(M, g_{ab})$ and $(M', g'_{ab})$ are {\em isometric} if there is a diffeomorphism $\varphi: M \rightarrow M'$ such that $\varphi_* (g_{ab})=g'_{ab}$. A spacetime $(M, g_{ab})$ is {\em extendible} if there exists a spacetime $(M', g'_{ab})$ and a (proper) isometric embedding $\varphi: M \rightarrow M'$ such that $\varphi(M) \subset M'$. Here, the spacetime $(M', g'_{ab})$ is an {\em extension} of $(M, g_{ab})$. A spacetime is {\em inextendible} if it has no extension.

We say a curve $\gamma: I\rightarrow M$ is not {\em maximal} if there is another curve $\gamma': I'\rightarrow M$ such that $I$ is a proper subset of $I'$ and $\gamma(s)=\gamma'(s)$ for all $s \in I$. A spacetime $(M, g_{ab})$ is {\em geodesically complete} if every maximal geodesic $\gamma: I \rightarrow M$ is such that $I=\mathbb{R}$. A spacetime is {\em geodesically incomplete} if it is not geodesically complete.

For any two points $p, q\in M$, we write $p \ll q$ if there exists a future-directed timelike curve from $p$ to $q$. We write $p<q$ if there exists a future-directed causal curve from $p$ to $q$. These relations allow us to define the {\em timelike and causal pasts and futures} of a point $p$: $I^-(p)=\{q:q \ll p\}$, $I^+(p)=\{q:p \ll q\}$, $J^-(p)=\{q:q<p\}$, and $J^+(p)=\{q:p<q\}$. Naturally, for any set $S \subseteq M$, define $J^+[S]$ to be the set $\bigcup \{J^+(x): x \in S\}$ and so on. A set $S \subset M$ is {\em achronal} if $S \cap I^-[S]=\varnothing$.

A point $p \in M$ is a {\em future endpoint} of a future-directed causal curve $\gamma: I \rightarrow M$ if, for every neighborhood $O$ of $p$, there exists a point $t_0 \in I$ such that $\gamma(t) \in O$ for all $t>t_0$. A {\em past endpoint} is defined similarly. A causal curve is {\em future inextendible} (respectively, {\em past inextendible}) if it has no future (respectively, past) endpoint. If an incomplete geodesic is timelike or null, there is a useful distinction one can introduce. We say that a future-directed causal geodesic $\gamma: I \rightarrow M$ without future endpoint is {\em future incomplete} if there is an $r \in \mathbb{R}$ such that $s<r$ for all $s \in I$. A {\em past incomplete} causal geodesic is defined analogously.

For any set $S \subseteq M$, we define the {\em past domain of dependence of $S$}, written $D^-(S)$, to be the set of points $p \in M$ such that every causal curve with past endpoint $p$ and no future endpoint intersects $S$. The {\em future domain of dependence of $S$}, written $D^+(S)$, is defined analogously. The entire {\em domain of dependence of $S$}, written $D(S)$, is just the set $D^-(S) \cup D^+(S)$. The {\em edge} of an achronal set $S \subset M$ is the collection of points $p\in S$ such that every open neighborhood $O$ of $p$ contains a point $q \in I^+(p)$, a point $r \in I^-(p)$, and a timelike curve from $r$ to $q$ which does not intersect $S$. A set $S \subset M$ is a {\em slice} if it is closed, achronal, and without edge. A spacetime $(M, g_{ab})$ which contains a slice $S$ such that $D(S)=M$ is said to be {\em globally hyperbolic} and the set $S$ is a {\em Cauchy surface}.

We define the {\em future Cauchy horizon} of $S$, denoted $H^+(S)$, as the set $\overline{D^+(S)}-I^-[D^+(S)]$. The {\em past Cauchy horizon} of $S$ is defined analogously. One can verify that $H^+(S)$ and $H^-(S)$ are closed and achronal. The {\em Cauchy horizon} of $S$, denoted $H(S)$, is the set $H^+(S) \cup H^-(S)$. We have $H(S)=\dot{D}(S)$ and therefore $H(S)$ is closed. Also, a non-empty, closed, achronal set $S$ is a Cauchy surface if and only if $H(S)=\varnothing$.

Now consider the triple $(\Sigma, h_{ab}, k_{ab})$. Here, $\Sigma$ is a connected manifold of dimension $n-1$, $h_{ab}$ is a Riemannian metric on $\Sigma$, and $k_{ab}$ is a symmetric field on $\Sigma$. Let $^{(n-1)}R$ be the scalar curvature of $h_{ab}$ and let $D_a$ be the unique derivative operator compatible with $h_{ab}$. We take $(\Sigma, h_{ab}, k_{ab})$ to be a (vacuum) {\em initial data set} if the following constraint equations are satisfied \citep[p.~259]{Wald}:

\begin{align*}
 ^{(n-1)}R-(k_a^{\ a})^2+k_{ab}k^{ab}&=0,\\
  D_bk_a^{\ b}-D_ak_b^{\ b}&=\textbf{0}.
\end{align*}

Let $(\Sigma, h_{ab}, k_{ab})$ be an initial data set. We call a spacetime $(M, g_{ab})$ a {\em maximal Cauchy development} of $(\Sigma, h_{ab}, k_{ab})$ if it has the following properties: (i) $(M, g_{ab})$ is a vacuum solution to Einstein's equation. (ii) $(M, g_{ab})$ is globally hyperbolic with Cauchy surface $\Sigma$. (iii) The induced metric and extrinsic curvature of $\Sigma$ are $h_{ab}$ and $k_{ab}$. (iv) Every other spacetime which satisfies (i)--(iii) can be mapped isometrically into a subset of $(M, g_{ab})$. Note that, by property (iv), a maximal Cauchy development of an initial data set $(\Sigma, h_{ab}, k_{ab})$ is unique. We can now state the following celebrated result. \\

\begin{prop}[\citet{CB+Geroch}] Let $(\Sigma, h_{ab}, k_{ab})$ be an initial data set with $\Sigma$ three-dimensional. There exists a maximal Cauchy development of $(\Sigma, h_{ab}, k_{ab})$.\end{prop}

So we have a clear sense in which the state of the universe at any particular ``time'' can be used to uniquely determine the state of the universe at all other ``times'' if attention is restricted to spacetimes that are four-dimensional vacuum solutions that are appropriately maximal. Now, given that \emph{every} two-dimensional spacetime is a vacuum solution, it should not be too surprising that the above proposition does not generalize. We have the following.\\

\begin{prop} Let $(\Sigma, h_{ab}, k_{ab})$ be an initial data set with $\Sigma$ one-dimensional. There is no maximal Cauchy development of $(\Sigma, h_{ab}, k_{ab})$.\end{prop}

\begin{proof} Let $(\Sigma, h_{ab}, k_{ab})$ be an initial data set where $\Sigma$ is one-dimensional. Let us proceed indirectly: Suppose there exists a maximal Cauchy development of $(\Sigma, h_{ab}, k_{ab})$ and let this two-dimensional spacetime be $(M, g_{ab})$. Let $O \subset M$ be any open set which is disjoint from $\Sigma \subset M$. Consider the spacetime $(M, g'_{ab})$ where $g'_{ab}=\Omega^2g_{ab}$ and $\Omega: M \rightarrow \mathbb{R}$ is a smooth, strictly positive scalar function which is chosen so that $\Omega(p)=1$ for all $p \in M-O$ and the spacetimes $(M, g_{ab})$ and $(M,g'_{ab})$ are not isometric.

By Prop.~\ref{EinsteinTensor} above, $(M,g'_{ab})$ is a vacuum solution. Since $(M, g_{ab})$ is globally hyperbolic with Cauchy surface $\Sigma$, and is conformally related to $(M,g'_{ab})$, the latter spacetime is globally hyperbolic with Cauchy surface $\Sigma$ as well. Because $g_{ab}=g'_{ab}$ on $M-O$, the induced metric and extrinsic curvature of $\Sigma \subset M-O$ in the spacetime $(M,g'_{ab})$ are $h_{ab}$ and $k_{ab}$ respectively. Thus, by the definition of maximal Cauchy development, $(M,g'_{ab})$ can be isometrically embedded into a subset $(M, g_{ab})$. But this is impossible since $(M, g_{ab})$ and $(M,g'_{ab})$ are not isometric.  \end{proof}

All by itself, the non-existence of maximal Cauchy developments in two-dimensional general relativity marks another significant break from the usual four-dimensional case; there is a kind of breakdown of determinism here that is not present in four dimensions. But there is an interesting corollary one finds as well: without maximal Cauchy developments, one loses an important tool commonly used to distinguish between ``physically reasonable'' and ``physically unreasonable'' models of general relativity. Take, for example, the ``cosmic censorship conjecture'' \citep{Penrose} which is the idea that all ``physically reasonable'' spacetimes are free of the ``ghastly pathologies of naked singularities'' \citep[p.~66]{Earman}. The physical formulation of one influential version of the conjecture is this: ``All physically reasonable spacetimes are globally hyperbolic'' \citep[p.~304]{Wald}.

To express the statement more precisely, we require a further definition.  We will say that a spacetime $(M,g_{ab})$ is \emph{strongly causal} if for any point $p\in M$, and any neighborhood $O$ of $p$, there exists a neighborhood $V\subseteq O$ of $p$ such that no causal curve intersects $V$ more than once.  The cosmic censorship hypothesis may then be expressed as follows \citep{Geroch+Horowitz, Wald}.\\

\begin{conj} Let $(\Sigma, h_{ab}, k_{ab})$ be an initial data set with $\Sigma$ three-dimensional. If the maximal Cauchy development of this initial data is extendible, for each $p \in H^+(\Sigma)$ in any extension, either strong causality is violated at $p$ or $\overline{I^-(p) \cap \Sigma}$ is noncompact.\footnote{If $\overline{I^-(p) \cap \Sigma}$ is noncompact, then $\Sigma$ is a poor choice of initial data set (e.g. the spacelike hyperboloid contained in the causal past of a point in Minkowski spacetime). See \citet[p.~76]{Earman}.} \end{conj}

\noindent Given that maximal Cauchy developments do not exist in two dimensions, how might one express (a version of) the cosmic censorship conjecture in a general way?

Let $(K,g_{ab})$ be a globally hyperbolic spacetime. Let $\varphi: K \rightarrow K'$ be an isometric embedding into a spacetime $(K', g_{ab}')$. We say $(K', g'_{ab})$ is an  {\em effective extension} of $(K ,g_{ab})$ if, for some Cauchy surface $S$ in $(K,g_{ab})$, $\varphi[K]$ is a proper subset of $int(D(\varphi[S]))$ and $\varphi[S]$ is achronal. Hole-freeness can then be defined as follows.\footnote{See \citet{GerochPrediction} for an earlier definition and \label{ManchakHoleFree} for a discussion of why a revision was needed.} A spacetime $(M,g_{ab})$ is {\em hole-free} if, for every set $K\subseteq M$ such that $(K, g_{ab})$ is a globally hyperbolic spacetime with Cauchy surface $S$, if $(K', g_{ab})$ is not an effective extension of $(K,g_{ab})$ where $K'=int(D(S))$, then there is no effective extension of $(K,g_{ab})$.

With this background, \citet[pp.~75--98]{Earman} suggests the following formulation of the cosmic censorship hypothesis.\\

\begin{conj} Let $(M,g_{ab})$ be an inextendible, hole-free, vacuum solution. If $S \subset M$ is a slice and there exists a $p \in H^+(S)$, then either strong causality is violated at $p$ or $\overline{I^-(p) \cap S}$ is noncompact. \end{conj}

\noindent This latter conjecture is, as far as we know, still open when $(M,g_{ab})$ is four-dimensional. But in fact that it is false in the two-dimensional case.\\

\textbf{Example.} Let $(M,\eta_{ab})$ be two-dimensional Minkowski spacetime.\footnote{That is, we suppose $M$ is diffeomorphic to $\mathbb{R}^2$ and $\eta_{ab}$ is flat and geodesically complete.} Let $q$ be any point in $M$ and $S \subset M$ be any slice such that $q \in I^+[S]$. Let $M'=M-\{q\}$ and consider a smooth, strictly positive scalar field $\Omega: M' \rightarrow \mathbb{R}$ that approaches infinity as the ``missing point'' $q$ is approached. Now consider the spacetime $(M', g_{ab})$ where $g_{ab}=\Omega^2\eta_{ab}$. Because the conformal factor $\Omega$ blows up, it renders the spacetime $(M', g_{ab})$ geodesically complete. It is thus hole-free and inextendible \citep{Manchak2014}. By Prop.~\ref{EinsteinTensor} above, the spacetime is a vacuum solution.  Now consider any point $p \in H^+(S)$. We know that strong causality is not violated at $p$ since the spacetime is stably causal (because the spacetime admits a global time function; see \citet[p.~199]{Wald}) and that $\overline{I^-(p) \cap S}$ is compact. (See figure \ref{fig:JB}) \\

\begin{figure}[h]    \centering
   \includegraphics[width=3in]{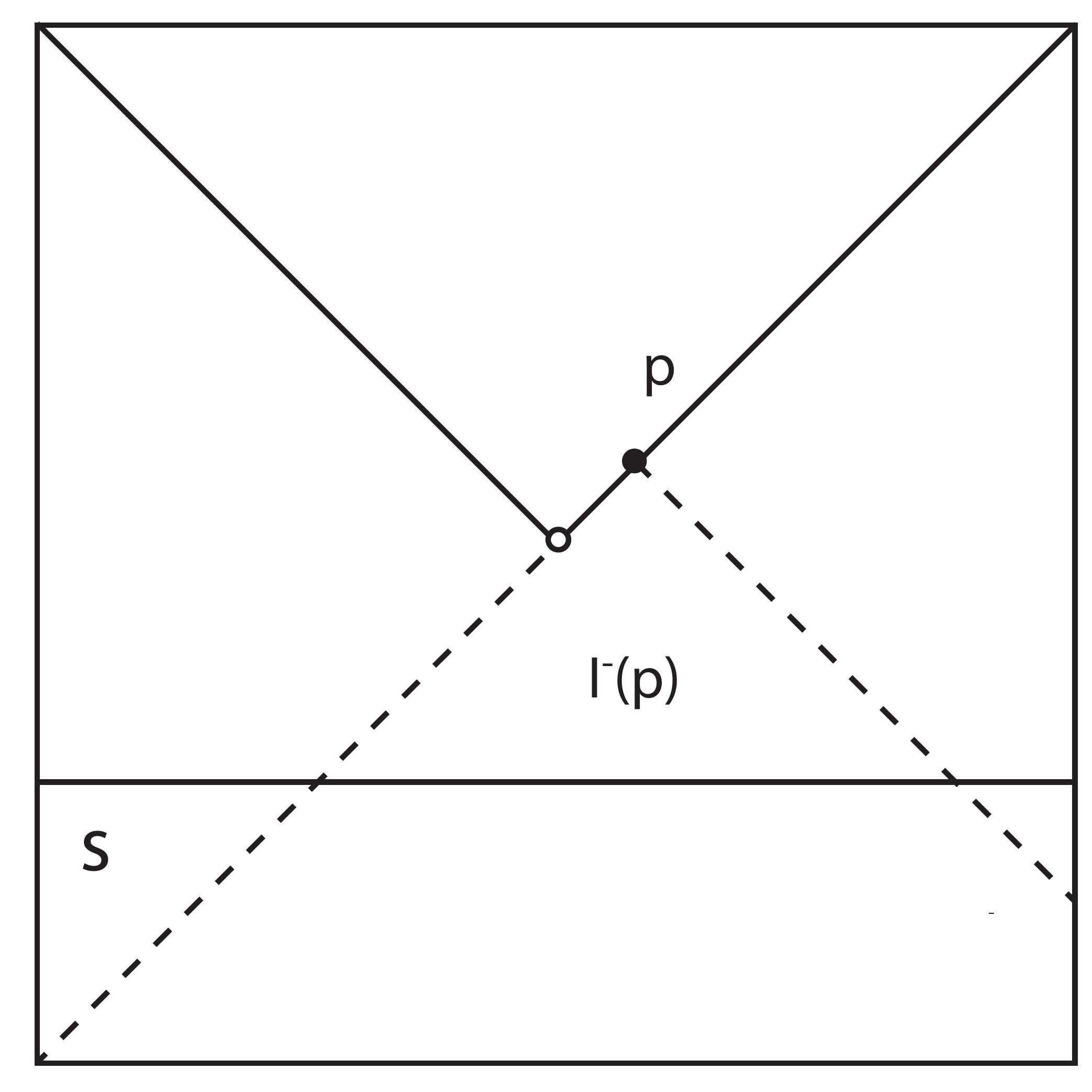}
   \caption{\label{fig:JB} The set $\overline{I^-(p) \cap S}$ is compact.}
\end{figure}

Does the above example show the two-dimensional version of the ``cosmic censorship conjecture'' to be false or does it merely suggest a variant formulation of the conjecture? This is not an easy question to pursue given the ``large and diverse class of ideas and motivations'' falling under the label ``cosmic censorship'' \citep[p.~99]{Earman}. One thing is clear, however: the example certainly demonstrates just how difficult it is to get a grip on any characterization of ``physically reasonable'' two-dimensional spacetimes.\footnote{These difficulties in two dimensions exacerbate the problem of characterizing ``physically reasonable'' spacetimes in any dimension. See \citet{Manchak2011} for details.} The construction of the example seems to be somewhat  ``artificial'' and yet it is about as locally well-behaved as one could demand---it is a vacuum solution, after all. Globally, the example also checks all the usual boxes required of ``physically reasonable'' spacetimes: it is stably causal, inextendible, and free of holes. Possibly the only count against the example is that it is geodesically complete while the singularity theorems of \citet{Hawking+Penrose} suggest that many ``physically reasonable'' spacetimes are geodesically incomplete. There are two separate responses to this line of reasoning, however.

First, one can easily consider a closely related example where the conformal factor $\Omega: M' \rightarrow \mathbb{R}$ goes to zero as the ``missing point'' $q$ is approached. The resulting spacetime $(M', \Omega^2\eta_{ab})$ is now geodesically incomplete but remains a counterexample to the conjecture.

The second response is of some independent interest: it is not clear that the singularity theorems are well-motivated in two dimensions. A crucial assumption for all of the major singularity theorems is the ``causal convergence condition,'' which is the requirement that $R_{ab}\xi^a\xi^b \geq 0$ for all causal vectors $\xi^a$ \citep{Hawking+Ellis, Senovilla}. Usually, this assumption is justified by mentioning that, in four dimensions, it is equivalent to the ``strong energy condition''---the requirement that $T_{ab}\xi^a\xi^b - \frac{1}{2}T \geq 0$ for all unit timelike vectors $\xi^a$.  It is often assumed that the strong energy condition is satisfied by all ``physically reasonable'' models of the universe \citep{Hawking+Ellis}.\footnote{Actually, there are good reasons to be skeptical that the strong energy condition is satisfied in all physically reasonable spacetimes; see \citet{CurielEC}.  Still, it is widely assumed and so its status is of interest here.}

We close this section by drawing attention to the fact that this equivalence between the conditions does not hold in two dimensions.

\begin{prop} In two dimensions, the causal convergence condition is not equivalent to the strong energy condition. \end{prop}

\begin{proof} Let $(M,\eta_{ab})$ be two-dimensional Minkowski spacetime where $M=\mathbb{R}^2$ and $\eta_{ab}=\nabla_at\nabla_bt-\nabla_ax\nabla_bx$, for standard coordinates $(t,x)$. Let $\Omega: M \rightarrow \mathbb{R}$ be the function $\Omega(t,x)=\exp(t^2)$. Consider the spacetime $(M, g_{ab})$ where $g_{ab}=\Omega^2 \eta_{ab}$. Because $(M, g_{ab})$ is two-dimensional, by Prop.~\ref{EinsteinTensor} above, we know (via Einstein's equation) that its associated $T_{ab}$ is the zero tensor. This clearly implies that $(M, g_{ab})$ satisfies the strong energy condition. But, using the fact that $(M,\eta_{ab})$ is flat, one can verify \citep[p.~466]{Wald} that $R_{ab}=-2\eta_{ab}=-2\exp(-t^2)g_{ab}$ where $R_{ab}$ is the Ricci tensor associated with $(M, g_{ab})$. Thus, if $\xi^a$ is a unit timelike vector at the point $(0,0) \in M$, we find $R_{ab}\xi^a\xi^b$ there to be $-2$. Thus, $(M, g_{ab})$ fails to satisfy the causal convergence condition. \end{proof}

\section{No Newtonian Limit}\label{sec:Newton}

We now consider the relationship between general relativity and Newtonian gravitation in two dimensions in light of Prop.~\ref{EinsteinTensor}.  We begin by reviewing some facts about the sense in which Newtonian gravitation is a limit of general relativity in four dimensions, and then show that the same limit does not obtain in two dimensions.

Recall that in four dimensions, one can present Newtonian gravitation as a theory set in a \emph{classical spacetime},\footnote{For details on the formulations of Newtonian gravitation discussed here, see \citet{Trautman} and \citet[Ch.~4]{MalamentGR}.} which is a quadruple $(M,t_a,h^{ab},\nabla)$, where $M$ is a connected 4-manifold of events, $t_a$ is a non-vanishing one-form on $M$, $h^{ab}$ is a smooth symmetric tensor field such that for all one-forms $\tau_a$ on $M$, $h^{ab}\tau_b = \mathbf{0}$ iff $\tau_a=\alpha t_a$ for some smooth scalar field $\alpha$; and $\nabla$ is a covariant derivative operator on $M$.  We assume that $\nabla$ is compatible with $t_a$ and $h^{ab}$, in the sense that $\nabla_a t_b=\mathbf{0}$ and $\nabla_a h^{bc}=\mathbf{0}$.  We will say that a vector $\xi^a$ is \emph{timelike} if $\xi^at_a\neq 0$; otherwise it is \emph{spacelike}; a curve is timelike (resp.~spacelike) if its tangent field is, everywhere.

In this framework, ``ordinary'' Newtonian gravitation is a theory in which $\nabla$ is flat, and gravitational effects are described using a gravitational potential $\varphi$ satisfying Poisson's equation, $\nabla_a\nabla^a\varphi = 4\pi\rho$, where $\nabla^a\varphi = h^{an}\nabla_n\varphi$ and $\rho$ is a scalar field representing the mass density in spacetime.  Given a solution to this equation, small bodies will accelerate according to the law $\xi^n\nabla_n\xi^a = \nabla^a\varphi$, where $\xi^a$ is the unit tangent to the center-of-mass worldline of the body.  One can also formulate Newtonian gravitation as a ``geometrized'' theory, on which spacetime is curved, with Ricci curvature satisfying $R_{ab}=4\pi\rho t_at_b$; on this theory, small bodies follow timelike geodesics.

In four dimensions, there is a precise sense in which the geometrized theory may be understood as a ``classical limit'' of general relativity, i.e., a limit in which $c\rightarrow \infty$, where $c$ is the speed of light.\footnote{For details on the sense in which this limiting procedure captures the idea that the speed of light is diverging, see \citet{FletcherThesis}.  Observe that our claim is that any sequence of models of general relativity satisfying the conditions described in the main text converge to models of Newtonian gravitation, and not that any model of general relativity has a natural or unique classical analogue, nor that every model of Newtonian gravitation arises as the limit of a sequence of relativistic spacetimes.}  In particular, suppose we have, on a fixed manifold $M$, a one-parameter family $\lambda\mapsto g_{ab}(\lambda)$ of Lorentz-signature metrics, for $\lambda$ defined on $(0,k)\subseteq \mathbb{R}$ with $k\geq 0$, a closed one-form $t_a$, and a smooth tensor field $h^{ab}$, all satisfying the following two conditions:\footnote{In addition, we assume all of the one-parameter families we consider are differentiable, to all orders, in $\lambda$.}
\begin{align*}
\lim_{\lambda\rightarrow 0}g_{ab}(\lambda) &= t_a t_b,\\
\lim_{\lambda\rightarrow 0}\lambda g^{ab} &=-h^{ab},
\end{align*}
where the limit is taken in the so-called $C^{\infty}$ point-open topology.\footnote{Here we follow \citet{Kunzle} and \citet{MalamentLimit}, but one can consider limits taken in other topologies; see \citet{FletcherThesis} for a discussion; nothing in what follows turns on the difference.}  Then we have the following result.
\begin{thm}[\citet{MalamentLimit}]\label{MalamentLimit} Let $M$ be a 4-manifold and let $g_{ab}(\lambda)$ be a one-parameter family of Lorentzian metrics satisfying the following conditions.
\begin{enumerate}
\item The family $g_{ab}(\lambda)$ converges to fields $t_at_b$ and $h^{ab}$ as just described; and
\item The Einstein tensor $G^{ab}(\lambda)$ converges to some smooth field $T^{ab}$.\footnote{One might be surprised by this condition.  \citet{MalamentLimit} assumes that there is a one-parameter family of fields $T^{ab}(\lambda)$ that converge to some smooth field $T^{ab}$; and he assumes that Einstein's equation holds for each $\lambda$, in the form $\overset{\lambda}{R}_{ab}=8\pi(T_{ab}(\lambda)-\frac{1}{2}\overset{\lambda}{T}g_{ab}(\lambda))$, where $\overset{\lambda}{T}$ is the trace of $T^{ab}(\lambda)$ and $\overset{\lambda}{R}_{ab}$ is the Ricci tensor associated with $g_{ab}(\lambda)$.  But in four dimensions, these two conditions are equivalent to requiring that the Einstein tensors determined by $g_{ab}(\lambda)$ converge to some smooth field $T^{ab}$.  Still, one might think that we should begin with this alternative formulation of Einstein's equation, which is only equivalent to Eq. \eqref{EinsteinEquation} in four dimensions.  But as we describe in sections \ref{sec:AltEq} and \ref{sec:EmpiricalContent}, in two dimensions the alternative version of Einstein's equation that Malament considers decouples curvature from energy-momentum, and so there is no chance that it will yield Poisson's equation in a classical limit.}
\end{enumerate}
Then there exist on $M$ a derivative operator $\nabla$ and a smooth field $\rho$, such that
\begin{enumerate}
\item $\overset{\lambda}{\nabla}\rightarrow \nabla$ as $\lambda\rightarrow 0$;
\item $(M,t_a, h^{ab},\nabla)$ is a classical spacetime satisfying $R^a{}_b{}^c{}_d = R^c{}_d{}^a{}_c$;
\item $G_{ab}(\lambda)\rightarrow \rho t_a t_b$ as $\lambda\rightarrow 0$; and
\item $R_{ab} = 4\pi\rho t_a t_b$, where $R_{ab}$ is the Ricci tensor determined by $\nabla$.
\end{enumerate}\end{thm}

We now consider the two-dimensional case. As we have characterized it here, one can certainly make sense of a notion of two-dimensional classical spacetime: one simply defines the structures described above on a 2-manifold of events rather than a 4-manifold.  (Note that in this case, no analogue to Prop.~\ref{EinsteinTensor} holds, because the derivative operator $\nabla$ is not determined by a non-degenerate metric.)  To be sure, Newtonian gravitational theory in this setting has some strange features: for instance, in the non-geometrized theory, the solution to a central force problem, for a point mass located at the origin (in an appropriate coordinate system $(t,x)$), yields a (discontinuous) gravitational potential
\[
\varphi = \begin{cases}-\frac{m}{2}x &\text{for } x >0,\\ \frac{m}{2}x &\text{for }x< 0,\end{cases}\]
whose associated gravitational field
\[
\nabla^a\varphi = \begin{cases}-\frac{m}{2}x^a &\text{for }x >0,\\ \frac{m}{2}x^a &\text{for }x< 0,\end{cases}
\]
is always attractive, but has magnitude independent of the distance from the central body.\footnote{This result is not so unexpected, perhaps: it arises for the same reason that in electrostatics (in three spatial dimensions), an infinite sheet with uniform charge density gives rise to an electric force that is independent of the distance from the sheet: it is a consequence of Gauss's law in one dimension.}
(Observe that this gravitational potential does not satisfy the boundary condition $\varphi\rightarrow 0$ as $|x| \to \infty$; moreover, $\varphi$ is generally unique only up to addition of a homogeneous solution to Poisson's equation, which in this setting takes the form $\varphi =A x$ for some constant $A$, and so the boundary condition cannot be realized for non-vanishing $\varphi$.) More generally, the force between two point masses $m$ and $m'$ will be given by $F^a = -mm'r^a$, where $r^a$ is the unit (spacelike) vector relating them.

But these strange features are not barriers to the theory being mathematically well-defined; moreover, general relativity is also strange in two (spacetime) dimensions.  The point we want to make now is that these two theories are strange in different, incompatible ways.  Perhaps most strikingly, in geometrized Newtonian gravitation in two dimensions, one \emph{can} have matter sources in the geometrized Poisson equation.  This is because the Ricci tensor associated with an arbitrary derivative operator compatible with classical metrics $t_a$ and $h^{ab}$ need not vanish---just as the Ricci tensor associated with an arbitrary Lorentzian metric need not vanish (even though the associated Einstein tensor always does vanish).

This difference between the theories leads us to the striking (further) observations that in two dimensions it is no longer the case that Newtonian gravitation is a ``classical limit'' of GR, at least not in the sense we described above in four dimensions.  To see why, suppose that we have a family of Lorentzian metrics $g_{ab}(\lambda)$ on a 2-manifold, converging to fields $t_at_b$ and $h^{ab}$ as above.  It immediately follows, from Prop.~\ref{EinsteinTensor}, that their associated Einstein tensors \emph{also} converge---and that they converge to $\mathbf{0}$.  But the Ricci tensor, $R_{ab}(\lambda)$ associated to each $g_{ab}(\lambda)$ need not vanish for any $\lambda$, and indeed, may not vanish in the limit---even when the derivative operators associated with $g_{ab}(\lambda)$ converge to a derivative operator $\nabla$ compatible with $t_a$ and $h^{ab}$.

To make this concrete, consider the manifold $\mathbb{R}^2$ with standard coordinates $t,x$.  In units where $c=G=1$, consider the following one-parameter family of metrics:
\[
g_{ab}(\lambda)=d_a t \; d_b t -\lambda(1+t^2)d_ax \; d_bx
\]
Clearly this family converges pointwise to $t_at_b$ as $\lambda\rightarrow 0$, where $t_a=d_a t$ is, by construction, a closed one-form; likewise we have $-\lambda g^{ab}(\lambda)\rightarrow h^{ab}=\frac{1}{1+t^2}\left(\frac{\partial}{\partial x}\right)^a\left(\frac{\partial}{\partial x}\right)^b$.  The derivative operator associated with $g_{ab}(\lambda)$ for each $\lambda$ may be written $\overset{\lambda}{\nabla}=(\partial,C^a{}_{bc}(\lambda))$, where $\partial$ is the coordinate derivative operator and $C^a{}_{bc}(\lambda) = \frac{t}{2}\lambda\left(\frac{\partial}{\partial t}\right)^ad_bx \; d_cx + \frac{2t}{1+t^2}\left(\frac{\partial}{\partial x}\right)^ad_{(b}x \; d_{c)}t$; this has associated Ricci tensor $R_{ab}(\lambda)=\frac{1}{2}Rg_{ab}(\lambda)$, where $R=-\frac{2}{(1+t^2)^2}$ has no dependence on $\lambda$.  One can easily confirm, then, that $\overset{\lambda}{\nabla}\rightarrow \nabla=(\partial,\frac{2t}{1+t^2}\left(\frac{\partial}{\partial x}\right)^ad_{(b}x \; d_{c)}t)$, satisfying $R^a{}_b{}^c{}_d=R^c{}_d{}^a{}_b$; and that $R_{ab}(\lambda)\rightarrow \frac{1}{2}Rt_at_b$.  It follows that $(\mathbb{R}^2,t_a,h^{ab},\nabla)$ is a classical spacetime, but with matter source $\rho=\frac{1}{8\pi}R\neq 0$, even though for every $\lambda>0$, $(M,g_{ab}(\lambda))$ is a vacuum spacetime.  It is in this sense that (geometrized) Newtonian gravitation is not the classical limit of general relativity in two dimensions: the limit does not preserve both sides of Einstein's equation.\footnote{In this example, the value of $\rho$ achieved in the limit is everywhere negative.  But this does not hold in all such examples---consider instead the family of metrics $g_{ab}(\lambda)=d_a t \; d_b t -t\; d_ax \; d_bx$ defined on $\mathbb{R}^2$ for $t>0$, for which we find, in the limit, $R=\frac{1}{2t^2}>0$, yielding positive $\rho$.}

In summary, we have proved the following result. (Compare with Theorem \ref{MalamentLimit}.)
\begin{prop} There exists a 2-manifold $M$ and a one-parameter family $g_{ab}(\lambda)$ of Lorentzian metrics on $M$ satisfying the following conditions:
\begin{enumerate}
\item The family $g_{ab}(\lambda)$ converges to fields $t_at_b$ and $h^{ab}$ as just described; and
\item The Einstein tensor $G^{ab}(\lambda)$ converges to some smooth field $T^{ab}$.
\end{enumerate}
Moreover, there exist on $M$ a derivative operator $\nabla$ and a smooth field $\rho$, such that
\begin{enumerate}
\item $\overset{\lambda}{\nabla}\rightarrow \nabla$ as $\lambda\rightarrow 0$;
\item $(M,t_a, h^{ab},\nabla)$ is a classical spacetime satisfying $R^a{}_b{}^c{}_d = R^c{}_d{}^a{}_b$; and
\item $R_{ab} = 4\pi\rho t_a t_b$, where $R_{ab}$ is the Ricci tensor determined by $\nabla$.
\end{enumerate}
But it is nonetheless \emph{not} the case that $G_{ab}(\lambda)\rightarrow \rho t_a t_b$ as $\lambda\rightarrow 0$.
\end{prop}

\section{Non-Zero Cosmological Constant and Matter}\label{sec:cosConst}

Thus far, we have considered a generalization of general relativity to two dimensions based on the assumption that Eq.~\eqref{EinsteinEquation} holds.  We have thus ignored the possibility of a cosmological constant term $\Lambda g_{ab}$ appearing in Einstein's equation.  We now consider what happens if we do include this term.  In other words, we now suppose that the equation
\begin{equation}\label{EECC}
R_{ab}-\frac{1}{2}g_{ab}R - \Lambda g_{ab} =8\pi T_{ab}
\end{equation}
holds in two dimensions, where $T^{ab}$ is understood once again to be the total energy-momentum density associated with matter and $\Lambda\neq 0$ is taken to be a constant of nature---i.e., it is a quantity that takes the same (unspecified) value at every spacetime point, and in all models of the theory.\footnote{Here we follow \citet[p.~561]{EarmanCC}, who distinguishes ``two senses in which the cosmological constant can be a constant: the capital $\Lambda$ sense, according to which $\Lambda$ is a universal constant, and the lower case $\lambda$ sense, according to which $\lambda$ is the same throughout spacetime but can have different values in different universes'', and then argues that it is the capital $\Lambda$ sense that is taken for granted in standard approaches to deriving Eq.~\eqref{EECC} from an action principle---note however that many such approaches are inequivalent in two dimensions \citep{Deser1996}.  See also \citet{Bianchi+Rovelli}.  We do not take a stand on whether one \emph{should} think of $\Lambda$ in this way, i.e., as a constant of nature, or perhaps as something that can vary from model to model, as in unimodular gravity.}  (Observe the sign of the $\Lambda$ term in Eq.~\eqref{EECC}, which arises because of our metric signature convention.)

The first thing to observe about this case is that Prop.~\ref{EinsteinTensor} continues to hold.  But it no longer implies that all spacetimes are vacuum solutions; instead, Eq.~\eqref{EECC} simplifies in two dimensions to
\begin{equation}\label{EECC2}-\frac{\Lambda}{8\pi} g_{ab}=T_{ab}.\end{equation}
Thus, if $\Lambda\neq 0$ we have a (non-degenerate) field equation relating the spacetime metric to energy-momentum.

Since $\Lambda$ is a constant with some fixed value, Eq.~\eqref{EECC2} asserts that the energy-momentum tensor associated with matter is always some (fixed) multiple of the metric $g_{ab}$.  This in turn implies that $T^{ab}$ is (necessarily) non-vanishing, which means that \emph{no} spacetimes are vacuum spacetimes, and indeed, $T^{ab}$ is nowhere vanishing.  It follows from this observation that one can have matter, but not (isolated) bodies.  It also implies that $T^{ab}$ is always constant, since $\nabla_a(\Lambda g_{bc})=\mathbf{0}$; that $T=-\Lambda/4\pi$, where $T$ is the trace of $T^{ab}$; and that $T_{ab} = \frac{1}{2}Tg_{ab}$.

Given these strong constraints on energy-momentum in two dimensions, one might wonder if there are any candidate matter fields whose stress-energies can satisfy them.  Indeed, perhaps as one would expect, for some standard systems of equations, there are no non-trivial solutions compatible with Eq.~\eqref{EECC2} in two dimensions.  For instance, one can consider solutions to the Einstein-Klein-Gordon equations, where the Einstein equation is understood as Eq.~\eqref{EECC2} and the Klein-Gordon equation has the same form as in four dimensions:
\begin{equation}\label{KGE}
\nabla_a\nabla^a\varphi + m^2\varphi =0,
\end{equation}
with the associated energy-momentum tensor given by
\[
T^{ab} = \nabla^a\varphi\nabla^b\varphi - \frac{1}{2}g^{ab}(\nabla^n\varphi\nabla_n\varphi - m^2\varphi^2).
\]
Taking the trace of both sides yields $T = m^2\varphi^2$, which can be constant only if $\varphi$ is constant.  But if $\varphi$ is constant, then by Eq.~\eqref{KGE}, $m^2\varphi=0$, and thus either $m=0$ or $\varphi=0$.  In either case, we find $T=0$, which is incompatible with Eq.~\eqref{EECC2} for $\Lambda\neq 0$.  (On the other hand, it seems that the $m=0$ Klein-Gordon equation admits constant, non-vanishing solutions whose energy-momentum tensor vanishes identically, and which thus solve the Einstein-Klein-Gordon equations in two dimensions \emph{without} cosmological constant.)

Still, it turns out that at least some matter fields with non-vanishing energy-momentum tensors \emph{can} be defined in two dimensions. Consider the case of electromagnetism with a perfect fluid source, assuming that analogs of the Maxwell equations for the Faraday (field strength) tensor $F_{ab}$  and charge-current density $J^a$  in four dimensions hold in two:
\begin{align}
\nabla_{[a}F_{bc]} &= \mathbf{0},\\
\nabla_a F^{ab}& = J^b.
\end{align}
Further suppose that the energy-momentum tensor associated with $F_{ab}$ in two dimensions is also the analog of its expression in four, and similarly for the energy-momentum tensor associated with the charged fluid:
\begin{align}
\overset{EM}{T_{ab}} &= F_{am}F^m{}_b + \frac{1}{4} g_{ab} (F_{mn}F^{mn}),\\
\overset{PF}{T_{ab}} &= \rho \eta_a \eta_b - p (g_{ab} - \eta_a \eta_b),
\end{align}
where $\eta^a$ is the four-velocity field of the fluid and $\rho$ and $p$ are its scalar mass-energy density and pressure, respectively.  (We assume, for simplicity, that the four-velocity field $\eta^a$ is defined [and non-zero] everywhere.)

It will be convenient to express the metric and the Faraday tensor in terms of $\eta^a$.
In particular, at least locally we can express
\begin{equation}
g_{ab} = \eta_a \eta_b - \chi_a \chi_b
\end{equation}
for a unit spacelike field $\chi^a$  orthogonal to  $\eta^a$, which in two dimensions is unique up to  a choice of sign.
This determines a volume element
\begin{equation}
\epsilon_{ab} = 2\eta_{[a}\chi_{b]},
\end{equation}
which is the unique 2-form on the manifold up to a multiplicative scalar field. So, one can express the Faraday tensor $F_{ab}$, which is anti-symmetric, as
\begin{equation}
F_{ab} = f\epsilon_{ab} = 2f\eta_{[a}\chi_{b]},
\end{equation}
where $f$ is a scalar field on $M$. Note in particular that
\begin{align}
F_{am}F^m{}_b &=  (f\epsilon_{am})(f\epsilon^m{}_b) = f^2(\eta_a \chi_m - \chi_a \eta_m)(\eta^m \chi_b - \chi^m \eta_b) \nonumber\\
&= f^2 (\eta_a \eta_b - \chi_a \chi_b) = f^2g_{ab},\\
F_{mn}F^{mn} &= -2 f^2.
\end{align}
Using these two facts, we can express the total energy-momentum tensor $T_{ab} = \overset{EM}{T_{ab}} + \overset{PF}{T_{ab}}$ as
\begin{align*}
T_{ab} &= F_{am}F^m{}_b + \frac{1}{4} g_{ab} (F_{mn}F^{mn}) + \rho \eta_a \eta_b - p (g_{ab} - \eta_a \eta_b)\\
&= f^2 (\eta_a \eta_b - \chi_a \chi_b) - \frac{1}{2}f^2(\eta_a \eta_b - \chi_a \chi_b) + \rho \eta_a \eta_b + p \chi_a \chi_b \\
&= \left( \rho + \frac{1}{2} f^2 \right) \eta_a \eta_b  +\left( p - \frac{1}{2} f^2 \right) \chi_a \chi_b.
\end{align*}
Equating this with the expression for the energy-momentum given by Einstein's equation, $T_{ab} = -(\Lambda/8\pi)g_{ab} = -(\Lambda/8\pi)(\eta_a \eta_b - \chi_a \chi_b)$, yields the following two equations in three variables ($\rho,p,f^2$):
\begin{align}
\rho &= -\frac{1}{2}f^2 - \frac{\Lambda}{8\pi}, \\
p &= \frac{1}{2}f^2 + \frac{\Lambda}{8\pi}.
\end{align}
Thus the energy density and pressure of the fluid are critically balanced ($p = -\rho$) and are quadratic in the ``magnitude'' of the Faraday tensor, offset by the cosmological constant.

There are a few special cases to note.
\begin{description}
\item[$\Lambda=0$:] When the cosmological constant vanishes, so does the energy-momentum tensor. Contrary to expectation (cf. \citet{Collas}), it is possible to have non-vanishing matter fields even when the total energy-momentum tensor vanishes, because it is possible for the contributions from the electromagnetic field and the perfect fluid to cancel each other exactly.  That said, in this case we have negative mass density and $p=-\rho$, which implies that the energy-momentum tensor associated with the perfect fluid does not satisfy the weak energy condition.  (Compare this case with the discussion of scalar fields above.)
\item[$f^2=0$:] Without electromagnetic fields, the energy density and pressure must be everywhere constant and balanced exactly by the cosmological constant: $\rho = - p = -\Lambda/8\pi$, and $\rho>0$ implies that $\Lambda < 0$.
\item[$p=0$:] If pressure vanishes, we must have $f^2 = -\Lambda/4\pi$, which implies that $\rho=0$ (and $\Lambda < 0$).  Therefore dust (whether charged or not) is impossible in two dimensions.  In this case, since $f$ is constant, Maxwell's equations imply that $J^b = \mathbf{0}$, yet any observer with four-velocity $\xi^a$ at a point measures a constant electric field $E^a = F^a{}_b \xi^b = f\sigma^a$, where $\sigma^a$ is a spacelike unit vector orthogonal to $\xi^a$.
(The magnetic field is undefined.) In such a model, charged test particles are a bit like Rindler observers in special relativity: they accelerate at a constant rate forever.

\end{description}

This example shows that there are some cases in which matter fields may be defined in two dimensions, both when $\Lambda\neq 0$ and otherwise.  But as we will now argue, even if we do have matter, the relationship between matter and geometry in two dimensions is strikingly different from in four dimensions.  In particular, there is a sense in which  matter in two dimensions does not (necessarily) ``gravitate''.  To see this, fix any flat metric $\eta_{ab}$ on a two-dimensional manifold $M$ (admitting some flat metric).  It follows that this metric is a solution to Eq.~\eqref{EECC2}, for $T_{ab} = -(\Lambda/8\pi) \eta_{ab}\neq\mathbf{0}$.  Thus, the presence of matter does not imply that the spacetime has non-vanishing Ricci curvature or non-vanishing curvature scalar.  Contrast this result with the standard claim in four dimensions that, if $T^{ab}$ satisfies the strong energy condition (discussed in section \ref{sec:detSing}) then gravity is attractive, in the sense that nearby geodesics tend to converge. In two dimensions, $T^{ab}$ may satisfy any energy condition at all, while the ``geodesic deviation'' of the spacetime, measuring the degree to which nearby geodesics accelerate relative to one another, vanishes identically \citep[cf.][\S 2.7]{MalamentGR}.

Before moving on, we note that the discussion of this section, concerning matter in two dimensions when $\Lambda\neq 0$, may strike some readers as very strange.  After all, although matter may be defined, its dynamics is so constrained as to barely deserve the name.  Moreover, although matter is (necessarily) present, as a consequence of Eq.~\eqref{EECC2}, the relationship between matter and geometry is much weaker than in four dimensions.  One might conclude from this discussion that a non-zero cosmological constant in two dimensions is unphysical, and so one ought to conclude that $\Lambda$ must be $0$ in two dimensions. (Of course, one might equally argue that even the $\Lambda=0$ case is unphysical.  But it seems to us that a theory with a large number of vacuum solutions is physically significant in a way that a theory with no vacuum solutions and necessary, unphysical matter is not.)

But if one finds this argument compelling, it has the following consequence.  We arrived at Eq.~\eqref{EECC2} by supposing that the correct form of Einstein's equation, irrespective of dimension, is given by Eq.~\eqref{EECC}, for some (fixed) value of $\Lambda\neq 0$. But from this perspective, the assignment of a value to $\Lambda$ is independent of the choice of dimensionality of the spacetime manifold.  This means that, if $\Lambda$ must be $0$ in some subset of the permissible models in the theory, then $\Lambda$ must be $0$ in every case, independent of the dimensionality of any of the particular models.

Of course, there are ways to respond to this argument without accepting the conclusion.  For instance, one could deny the premise with which we began this section, that the cosmological constant should be taken to have the same value in all models of the theory.  Perhaps one could allow it to vary generally; or one could imagine it is the same in all models of a given dimension, but varies between dimensions.  One could also argue along the lines of what we will present in section \ref{sec:AltEq}, that one should not generalize general relativity to two dimensions by assuming Einstein's equation holds, with or without cosmological constant.  But any of these responses raises new questions concerning how we are to determine what can and cannot vary when one considers the space of possibilities according to a theory.  One principled answer, which we have adopted here for the sake of argument, is that constants of nature never vary across models; other answers are possible, but presumably require defense.

\section{On the Choice of Generalizing General Relativity via Einstein's Equation}\label{sec:alternatives}

The laundry list of unusual features of the version(s) of general relativity in two dimensions that we have been considering draw attention to the assumption with which we began, that Einstein's equation in two dimensions is the same as in four dimensions.  But should we accept that assumption?  Why should the two-dimensional equation have the same \textit{syntactic form} as the four-dimensional equation?  What kind of inductive evidence could we have to secure such an inference?

One line of evidence comes from various theoretical results establishing conditions under which the Einstein equation (or perhaps a class of equations including it) is uniquely determined as the field equation connecting geometry and matter in a spacetime theory.
We review the bearing of these on two-dimensional gravity in section \ref{sec:lovelock}.
These theorems have assumptions that could be questioned, of course,  so in section \ref{sec:AltEq} we describe how their conclusions could be and have been evaded.
Some of the alternative field equations proposed are justified in entirely different ways.
Finally, in section \ref{sec:EmpiricalContent}, we consider how any of these proposals could make contact with our empirical evidence for the four-dimensional theory, raising the possibility that what the ``correct'' two-dimensional theory is has no answer  and consequently that there is a plurality of two-dimensional theories that are viable in different contexts.

\subsection{Lovelock Variations}
\label{sec:lovelock}
Why ought the Einstein equation be the appropriate field equation for general relativity even in four dimensions?  There is a long history of attempts to provide an axiomatic or principled justification beyond Einstein's heuristic reasoning. These approaches have generally been founded on the following assumptions:\footnote{See \citet[\S17.5]{MTW} for a discussion of these approaches pre-Lovelock and some other, more heuristic approaches to determining the Einstein equation.}
\begin{enumerate}
\item  The field equation must take the form
\begin{equation}
\tilde{G}_{ab}(g) = T_{ab},
\label{eq:standard form}
\end{equation}
where $T_{ab}$ is the usual energy-momentum tensor field (although not necessarily assumed to be symmetric) and $\tilde{G}_{ab}$ is some (0,2)-tensor field whose value at a point depends only on the metric and its derivatives at that point.
\item The conservation condition holds, i.e., in light of the first condition, $\nabla_a\tilde{G}{}^{ab}(g)=\mathbf{0}$.
\end{enumerate}
The first condition ensures that the field equation is defined pointwise by the metric $g_{ab}$ and its derivatives, and that it equates the ``marble'' of geometry with the ``wood'' of matter.
The second condition demands that the conservation condition follows from the form of $\tilde{G}_{ab}$ alone.
Together, they do not imply that $\tilde{G}_{ab} \propto G_{ab} - \Lambda g_{ab}$, but with a few extra conditions, they do \citep{Vermeil1917,Cartan1922,Weyl22}:
\begin{enumerate}[resume]
\item $\tilde{G}_{ab}(g)$ depends only on the metric $g_{ab}$ and its first and second derivatives. \label{item:second-order}
\item $\tilde{G}_{ab}(g)$ is linear in the second derivatives of the metric $g_{ab}$. \label{item:linearity}
\item $\tilde{G}_{ab}(g)$ is symmetric. \label{item:symmetry}
\end{enumerate}
Note that there is no assumption concerning the dimensionality of spacetime.
One can maintain the same conclusion, however, by dropping conditions \ref{item:linearity} and \ref{item:symmetry} and adding such an assumption:
\begin{enumerate}[resume]
\item Spacetime is four-dimensional.
\end{enumerate}
This celebrated result, due to \citet{Lovelock71,Lovelock72}, is often quoted as one of the strongest foundations for the Einstein equation in four dimensions, yet it can be applied to two dimensions as well.
For, \citet{Lovelock71} provides a \textit{general} form for tensor fields $\tilde{G}_{ab}(g)$ satisfying conditions 1--\ref{item:second-order} and \ref{item:symmetry} that yields a unique answer in the two-dimensional case, too:\footnote{See also \citet{Navarro+Navarro} for a simplified proof connecting these results with the geometric concept of \textit{natural} tensor fields.  Although Lovelock's theorem places strong constraints on the form of gravitational field equations (satisfying the conditions described) in lower dimensions, it is more permissive in higher dimensions, which has led many theorists to focus on so-called "Lovelock theories" in higher dimensions, i.e., theories with field equations distinct from Einstein's equation but which satisfy the conditions of Lovelock's theorem.  Thus it is tempting to think that in higher dimensions, but not lower dimensions, one can find various theories with similar syntactic and semantic features to general relativity in four dimensions, and that this leads to an important distinction between the two cases.  But the arguments we give in section \ref{sec:AltEq} suggest that it is too fast, because insofar as one can question the assumptions of Lovelock's theorem in lower dimensions, one can also question them in higher dimensions.  It seems to us, then, that more needs to be said even in the higher dimensional case to justify a particular choice of generalization of Einstein's equation.}
\begin{equation}
\tilde{G}_{ab}(g) = Ag_{ab},
\label{eq:2D form}
\end{equation}
for some constant $A \in \mathbb{R}$.
Recently \citet{Navarro14} has shown how to achieve the same conclusions without condition \ref{item:symmetry}, so we may state the most general conclusion about two-dimensional general relativity thus:
\begin{prop}
Any divergence-free tensor field $\tilde{G}_{ab}$ in two dimensions naturally definable pointwise from the metric $g_{ab}$ and its first and second derivatives must take the form of Eq.~\eqref{eq:2D form}.
\end{prop}

A variation on Lovelock's approach, first described by \citet{Aldersley} and elaborated by \citet{Navarro+Sancho}, drops condition \ref{item:second-order} for the following ``dimensional analysis'' condition:
\begin{enumerate}[resume]
\item $\tilde{G}_{ab}(g)$ is \textit{independent of the unit of scale}, i.e., for any $\lambda > 0$, $\tilde{G}_{ab}(\lambda^2 g) = \tilde{G}_{ab}(g)$.
\label{item:scale-free}
\end{enumerate}
In other words, condition \ref{item:scale-free} states the invariance of $\tilde{G}_{ab}$ under homothetic transformations of the metric.
In fact, this assumption proves a conclusion somewhat stronger:
\begin{prop}
Any divergence-free tensor field $\tilde{G}_{ab}$ in two dimensions naturally definable pointwise from the metric $g_{ab}$ and its derivatives must take the form $\tilde{G}_{ab} = A G_{ab} (=\mathbf{0})$.
\end{prop}
Thus, condition \ref{item:scale-free} \textit{rules out} the cosmological constant term from appearing in the ``marble'' geometry of the field equation.
Does this imply that the proposition's assumptions are in conflict with the observed non-zero value of $\Lambda$?
Both \citet{Aldersley} and \citet{Navarro+Sancho} counsel that it does not insofar as it forces one to place that term on the right-hand side of Eq.~\eqref{eq:standard form}, strongly suggesting that $\Lambda$ be interpreted as a \textit{material} contribution to the equation rather than a \textit{geometrical} one.
This counsel applies equally in two dimensions as it does in four.

The strength of this result clearly depends on condition \ref{item:scale-free}; how should we interpret it?
A homothetic transformation $g_{ab} \mapsto \lambda^2 g_{ab}$, it is often said, ``amounts to a change in the time unit,'' so that condition \ref{item:scale-free} means that $\tilde{G}_{ab}$ (hence $T_{ab}$) is ``independent of [the choice of] the time unit'' \citep[\S6]{Navarro+Sancho}, but such an interpretation must be handled with care.
While it is true that a homothetic transformation effectively multiplies the lengths of all timelike curves by a constant factor, the resulting spacetime
may not necessarily be interpreted as simply the same as before but with different units.
This can occur if the equations relating fields on the spacetime contain dimensional (i.e., not purely numerical) constants that set an absolute scale for dimensional quantities \citep[pp.~372--3]{Aldersley}.
Thus a better description of condition \ref{item:scale-free} is that the field equation is temporally (and, from the constancy of the speed of light, spatially) scale invariant: changes of temporal (or spatial) scale are dynamical symmetries in the sense that they leave the field equation invariant.

In light of this, it is easy to see how the inclusion of the cosmological constant term breaks this symmetry, as a non-zero value thereof introduces a time (and therefore length) scale into the theory.
This is distinct from Newtonian gravitation, which does not have any such absolute scale.
It must be remarked, however, that it is difficult to see how scale invariance could be an \textit{a priori} condition on a theory of gravitation.
Whether it holds seems rather to be an empirical matter.\footnote{
  This is not to say that it is an \textit{implausible} assumption.
  If one believes, in some sense, geometry to be prior to matter, one might argue that the introduction of an absolute length or time scale arises only through the peculiar particularities of specific types of matter, so that in fact $\tilde{G}_{ab}$ should be scale invariant.
  However, it is still hard to see how this argument establishes anything more than the plausibility of condition \ref{item:scale-free}.
}

\subsection{Alternative Field Equations}
\label{sec:AltEq}
What the approaches described in section \ref{sec:lovelock} all have in common are the first two conditions: the form of the field equation \eqref{eq:standard form} and the conservation condition.
It is not difficult to find arguments for the conservation condition (e.g., see \citet[\S17.2]{MTW}).
Yet the assumption that the connection between geometry and matter must take the form prescribed in Eq.~\eqref{eq:standard form} is hardly ever questioned explicitly in the literature.  Why should the relationship between geometry and matter be so simply expressed?

For instance, in four dimensions the usual form of Einstein's equation is equivalent to
\begin{equation}
R_{ab} = 8 \pi \left( T_{ab} - \frac{1}{2} g_{ab} T \right) - \Lambda g_{ab},
\label{eq:einstein-alt}
\end{equation}
where $T=T^a{}_a$ is the trace of the energy-momentum tensor.  To show this, beginning with either equation, one	can take the trace of both sides to establish that $R = -8 \pi T - 4 \Lambda$, and substituting $R$ for $T$ (or vice versa) in the relevant place in the equation.
But this argument involves taking the trace of the metric, the value of which is the dimensionality of spacetime, so the equivalence depends on the assumption that spacetime is four-dimensional.   In any other dimension, including two, it does not hold.

Consequently, one could instead take the form given in Eq.~\eqref{eq:einstein-alt}, instead of the usual form, to be the field equation connecting geometry and matter in two dimensions.  But in this case, Prop.~\ref{EinsteinTensor} still holds, so we have that
\[
\frac{1}{2} R g_{ab} = 8\pi \left( T_{ab} - \frac{1}{2}g_{ab}T \right) - \Lambda g_{ab},
\]
whose trace yields that $R = -2\Lambda$.
Thus, this alternative version of Einstein's equation would lead us to conclude that
\begin{equation}
T_{ab} = \frac{1}{2} T g_{ab}.
\label{eq:energy-alt}
\end{equation}
In this theory, curvature and matter play roles converse from those reached by beginning with the other form of Einstein's equation: instead of the Ricci tensor being proportional to its trace times the metric and the energy-momentum proportional to the cosmological constant times the metric, the energy-momentum is proportional to its trace times the metric and the Ricci tensor is proportional to the cosmological constant times the metric.
In other words, this equation yields a universe of constant curvature whose geometry is totally decoupled from the energy of ponderable matter.

Unlike with the usual form of Einstein's equation, Eq.~\eqref{eq:einstein-alt} does \textit{not} guarantee that the conservation condition $\nabla_a T^{ab} = 0$ holds (except in four dimensions).
If one adds this as a separate field equation, then from Eq.~\eqref{eq:energy-alt} one immediately derives that $\nabla_a T = \mathbf{0}$, i.e., $T$ must be constant.
Thus the conservation condition forces the energy-momentum to take the same form as before, except the role played by the cosmological constant is now played by the (constant) trace of the energy-momentum.
In a sense, the resulting theory allows for strictly \textit{fewer} possibilities than before: with the analog of the original form of Einstein's equation, solutions were parameterized by a single real scalar field (the scalar curvature $R$), whereas with Eq.~\eqref{eq:einstein-alt} and the conservation condition, they are parameterized by a real number (the trace $T$ of the energy-momentum).\footnote{If one allows $\Lambda$ to vary between models, then in both cases one also may parameterize models by values of the cosmological constant.  But that is a real number in both cases, and so it does not change the analysis.}

Is there a different field equation that escapes the form \eqref{eq:standard form} and yet provides possibilities that seem more physically reasonable than the ones considered above in two dimensions?
In fact, the very considerations arising from Prop.~\ref{EinsteinTensor}---on the one hand, the lack of constraint on geometry, and on the other, the severe constraint on matter\footnote{An historical aside: as \citet[p.~72]{Wald} notes, Einstein came to reject an earlier field equation in which $T_{ab} \propto R_{ab}$ precisely because it demanded that both $R$ and $T$ are constant.}---have led physicists, starting with \citet{Teitelboim83,Teitelboim84} and \citet{Jackiw84,Jackiw85}, to propose instead the field equations\footnote{Caution: there is some variation in the literature on the choice of numerical coefficients and whether to exclude the cosmological constant or matter. The presentation here essentially follows that of \citet{Christensen+Mann}.}
\begin{align}
R - \Lambda &= 8\pi T, \label{eq:TJ theory} \\
\nabla_a T^{ab} &= \mathbf{0}.\label{eq:conservation}
\end{align}
Notably, the conservation condition is assumed separately from the field equation, which, though it equates geometry with energy-momentum, does so through scalar quantities instead of (0,2)-tensors.
Unlike higher-dimensional cases, however, in which this would underdetermine both geometric and matter degrees of freedom, in two dimensions the only degree of freedom in the geometry comes through the scalar curvature.

Eq.~\eqref{eq:TJ theory} is not obviously syntactically related to Einstein's equation in four dimensions.
Indeed, as \citet[p.~320]{Boozer} points out, the field equations are analogous to those of Nordstr\"om's 1913 theory \citep[p.~429]{MTW},
\begin{gather*}
R = 24\pi T,\\
C_{abcd} = \mathbf{0},
\end{gather*}
since, as we argued before, the only candidates for the Weyl tensor $C_{abcd}$ in two dimensions must vanish.
What would make Einstein's theory in two dimensions essentially the same as a clearly distinct theory in four?

One can adduce three sorts of arguments.
The first and most commonly argued position is that the \textit{qualitative} similarities between the \textit{solutions} of Eqs.~\eqref{eq:TJ theory} and \eqref{eq:conservation} and the four-dimensional Einsteinian theory are the relevant factors for comparison, not the language in which the equations are described.
Such claimed features include \citep{Sikkema+Mann,Christensen+Mann}:
\begin{itemize}
\item being derivable from a local action principle,
\item a Newtonian limit,
\item Robertson-Walker cosmological solutions,
\item gravitational waves, and
\item the gravitational collapse of dust into a black hole with an event horizon analogous to that of the four-dimensional Schwarzschild solution.
\end{itemize}
Hence, this theory's advocates have (if only implicitly) suggested that it is not the syntactic features of the field equation but qualitative semantic features---those of its models---that provide evidence that it is a relevant analog of general relativity in two dimensions.\footnote{
	\citet{Boozer} provides two ``derivations'' of the theory, one from a two-dimensional Newtonian theory and another using a ``principle of equivalence'' argument assuming that gravitation is represented by a scalar field.
    The former shows a sense in which Eq.~\eqref{eq:TJ theory} is a relativistic analog of the two-dimensional Poisson equation, while the latter results only in an ``effective'' geometry overlaid on an undetectable Minkowski background.
    These are fine as heuristic or motivational arguments, but because of the many assumptions and unforced choices made in the course of their development, they do not provide, in our opinion, evidence beyond those of the qualitative features already established.
}
But this evidence does not extend to any claim of uniqueness.

We have found further two arguments in the literature for some sort of uniqueness, but neither is conclusive.
The first, as developed by \citet{Mann92}, begins by pointing out that the usual Einstein equation in $n$ dimensions is equivalent to the following two equations, representing its trace and trace-free parts, respectively:
\begin{align}
\left(1 - \frac{n}{2} \right) R &= 8\pi G_n T,\\
R_{ab} - \frac{1}{n}R g_{ab} &= 8\pi G_n \left(T_{ab} - \frac{1}{n}T g_{ab} \right),
\end{align}
where Newton's constant $G_n$ is now assumed to depend on $n$.
In particular, if one assumes that $G_n/(1-n/2)$ is well-defined, non-zero, and $n$ can be treated as a continuous parameter, then one can define
\begin{equation}
G'_2 = \lim_{n \to 2} \frac{G_n}{1-n/2}
\end{equation}
as a kind of ``renormalized'' gravitational constant.
In this case the trace equation yields (the cosomological constant-free version of) Eq.~\eqref{eq:2D form}, and the trace-free equation becomes a mathematical identity.

The difficulty with this argument, aside from the mathematically dubious treatment of $n$ as a continuous parameter, is that it is not obvious how to justify that Newton's constant should depend on dimensionality, much less in precisely the way that makes $G'_2$ well-defined and non-zero.
If this dependence is not of a very particular form then one could well arrive at different field equations.\footnote{Note, too, that one would need to take a constant of nature---Newton's constant---to vary with dimension.  Compare with the discussion at the end of section \ref{sec:cosConst}.}

The second argument, described by \citet{Lemos+Sa}, also proposes to derive Eq.~\eqref{eq:TJ theory} from a limiting procedure.
Their argument is that because there is a case for general relativity to be a limiting case of Brans-Dicke theory in dimensions greater than two, whatever theory results from the same limit for two-dimensional Brans-Dicke theory ought to be considered the two-dimensional analog of GR.
Brans-Dicke theory is an alternative gravitational theory involving a scalar field $\phi$, interpreted as a kind of variable gravitational constant, and a new dimensionless constant, the Brans-Dicke constant $\omega$, which mediates the strength of the coupling between the variability of $\phi$ and matter.\footnote{Specifically, $ \Box \phi = [8\pi/(3+2\omega)]T$, where $\Box$ is the Laplace-Beltrami operator.}
There are arguments to the effect that taking $\omega \to \infty$ results in GR, and \citet{Lemos+Sa} suggest to extend this inductively to the two-dimensional case.

However, to achieve this result, the authors must assume that the ``effective'' field $\phi$, cosmological constant $\Lambda'$, and energy-momentum trace $T'$ are actually functions of $\omega$ of the following forms:
\begin{align}
\phi &= \phi_0 + \frac{\varphi}{4\omega} + O(\omega^{-2}),\\
\Lambda' &= \Lambda_0 + \frac{\Lambda}{2\omega} + O(\omega^{-2}),\\
T' &= T_0 - \frac{T}{4e^{2\phi_0}\omega} + O(\omega^{-2}),
\end{align}
where $\phi_0,\Lambda_0,T_0$ are real constants and $\varphi$ is a scalar field.\footnote{We have adjusted the notation somewhat to align it with that used so far.}

Clearly some of the difficulties with this argument are similar to those for the argument described by \citet{Mann92}.
As before, it is difficult to see why the ``renormalized'' field $T'$ and parameter $\Lambda'$ should depend on $\omega$ in the way they must to arrive at the desired conclusion.
These dependencies can be easily changed to arrive at different field equations.
(Unlike in typical applications of running constants or effective fields depending on, say, the energy scale, the parameter $\omega$ is constant: it cannot vary from context to context in a model.)
Moreover, the security of the premises that the $\omega \to \infty$ limit of Brans-Dicke should be always identified as GR, regardless of dimension, presumes dubiously that our evidence for the correct form of that theory is better than for GR.
Finally, besides questioning the cogency of the result in general on conceptual or mathematical grounds \citep{Faraoni,Bhadra+Nandi,Chauvineau}, one simply is not logically compelled to accept the inference from the $n \geq 3$ case to the $n=2$ case.

\subsection{The Empirical Content of Low-Dimensional Gravity}
\label{sec:EmpiricalContent}

The previous subsections described a few different approaches to justifying what, exactly, the field equation(s) in two-dimensional general relativity ought to be.
But we have argued that none of the arguments for uniqueness succeed without questionable premises, and so the question of what, exactly, two-dimensional general relativity is supposed to be seems not yet to have a conclusive answer, if that question is well-posed at all.
If there is any theory which deserves to be called \textit{the} two-dimensional version of GR, there needs to be a principled way in which that theory is supported.
Are there other considerations that might be brought to bear?

As suggested at the beginning of this section, one might try to determine the relevant field equation, or at least constraints on what it could be, through the empirical evidence we have for four-dimensional GR.
What relationship, though, does evidence for the four-dimensional theory have with a two-dimensional theory?
Since it is less obviously tasked with accurate or otherwise successful description, to what criteria must a two-dimensional theory be held?
Such questions are, perhaps surprisingly, rarely addressed in the literature on two-dimensional gravity.
One exception is \citet{Jackiw85}, who concedes that study of two-dimensional gravity probably has only pedagogical value.
His primary interest in low-dimensional gravity, in other words, is as a toy model for the real case of interest, four (or higher) dimensions, in which various calculations and ideas (especially pertaining to quantization) can proceed less encumbered from the complexities that more dimensions introduce.\footnote{
For more on the use of toy models in physics, see \citet{Hartmann1995}, \citet{Marzuoli2008}, \citet{LUCZAK20171}, and \citet{Toy}.}
Such strategies, after all, have been successful for other branches of study, such as condensed matter physics.
Consequently, if one has those sorts of goals in mind in investigating two-dimensional gravity, then the qualitative features of the solutions to Eq.~\eqref{eq:TJ theory} would override considerations coming from Lovelock's theorem and its variants.

\citet[p.~344]{Jackiw85} does consider two further possibilities, which are that lower-dimensional theories could have a kind of duality with higher-dimensional theory, and that lower-dimensional theories could describe the behavior of configurations of matter in four dimensions that are confined to move in fewer dimensions.
He dismisses them as either physically nonsensical or speculative, but this seems too quick to us.
Falling within the former sort of case, at least considered broadly, are situations in which the symmetries of a spacetime reduce the degrees of freedom to those of a two-dimensional model.\footnote{The literature on dimensional reduction is voluminous, not least because of the polysemy of the term. For a small sample of the literature on the sense in which it is used here, see \citet{Cadoni+Mignemi}, \citet{Kiem+Park}, and \citet{Schmidt}.
For a comparison with the method of \citet{Mann92} discussed above, see \citet{Mann+Ross}. }
Falling within the latter are models in which matter is (at least as a sufficiently good idealization) confined to an embedded two-dimensional Lorentzian submanifold.
In either case, evidence for the two-dimensional theory could be inherited from evidence for four-dimensional GR, simply because the two-dimensional theory describes a part of a world like ours.

One difference, however, is that there is no guarantee in these cases that there be a \textit{single} two-dimensional theory that deserves to be called the rightful analog of GR.
Indeed, perhaps we must be prepared to admit that multiple theories of two-dimensional general relativity may be viable.
This plurality of theories need not be problematic, however, as it would be if there were no good reasons to pick amongst the different versions yielding conflicting descriptions of phenomena.
Rather, different two-dimensional theories could apply in different contexts, depending on how that context is related to the more familiar four-dimensional theory for which we have more direct empirical evidence.\footnote{See, however, \citet{Fletcher2017} for arguments for the same pluralistic conclusion for four-dimensional GR. If distinct versions of four-dimensional general relativity are confirmed, the various two-dimensional theories derived from them through dimensional reduction could also be confirmed. In this sense, a plurality of four-dimensional theories does not undermine the empirical basis for a plurality of two-dimensional theories.}
A thorough investigation of what these theories can be, and the contexts in which they arise, must be left to future research.

\section{Conclusion}

We have described various features of general relativity in two dimensions, on several different ways of understanding what that theory should consist in.  As we have argued, general relativity in two dimensions is strikingly different from the theory in four dimensions.  On one natural way of understanding the two-dimensional theory, several of the characteristic features of the four-dimensional theory---such as the existence of an initial-value formulation, a well-defined Newtonian limit, and a dynamical dependence of spacetime on the presence of matter---do not appear to hold in two dimensions.  Alternative versions of the two-dimensional theory, meanwhile, eschew Einstein's equation.  Given this, it seems one needs to either qualify one's assertions concerning the features of general relativity, or else conclude that there is some sense in which general relativity requires spacetime to have four dimensions (or at least, have dimension greater than two).  As we argue in section \ref{sec:cosConst}, if one adopts the first option, there are consequences for what value the cosmological constant can take---at least if one adopts the view that the cosmological constant is a constant of nature, in the sense described by \citet{EarmanCC}.

We conclude by observing that the sort of extended reflection on a physical theory in other dimensions as presented here---which, we emphasize, is hardly unusual in the physics literature---raises important questions for philosophers of physics (and others) concerning how we understand the space of physical possibilities.  Briefly, many philosophers of physics would like to take the space of models of our physical theories as characterizing a space of ``physically possible worlds'' \citep[Ch.~2]{RuetscheIQT}.  We do not wish to take a stand on whether this is always the best way of understanding the notion of possibility captured by physical theorizing.
When one uses theory to describe physical toy models, for example, those models are not intended to describe  physical possibilities directly so much as to serve as tools for analogical reasoning about physical possibilities, among other uses; as discussed in section \ref{sec:EmpiricalContent}, this is the explicit attitude of \citet{Jackiw85} towards two-dimensional gravitational models.

But if one does think that physical theorizing is intended to capture the space of physical possibilities, then it seems that there is a certain indeterminacy in what this space consists in.
In particular, either it is possible that the world had a different number of dimensions than we observe, or not.  In the former case, the space of physically possible worlds presumably includes worlds with different numbers of dimensions.  But then, one might think that there must be some fact about the laws in those worlds.  Can we know what the laws are?  Should we be committed to the idea that the laws in other dimensions take the same form as in four dimensions, or should other considerations, such as we discuss in section \ref{sec:alternatives}, enter into our deliberations?  Do constants of nature vary among these possibilities, and if so, do they vary only when one changes dimension?  Should we be satisfied if, having chosen some way of generalizing some theory, the possibilities in other dimensions turn out to be qualitatively different from those in four dimensions?

In the latter case, meanwhile, where it is \emph{not} possible for the world to have been two-dimensional, it seems that there can be no fact of the matter about whether analyses of models with other dimensions (as happens in both quantum gravity and mathematical quantum field theory) track anything about the world at all, because there are no physical possibilities to which these models correspond.

\section*{Acknowledgments}

This paper is partially based upon work supported by the John Templeton Foundation grant ``Laws, Methods, and Minds in Cosmology".  It was originally conceived when three of us (Fletcher, Manchak, and Weatherall) were students in David Malament's general relativity course, in a previous Aeon.  We are grateful to David for many discussions related to this material, and to Brian Pitts and two anonymous referees for helpful comments on a previous version.  Weatherall wrote a substantial portion of this paper while a Visiting Fellow at the Australian National University; he is grateful for their support.  Fletcher acknowledges the support of the European Commission through a Marie Curie Fellowship (PIIF-GA-2013-628533) during the writing of this paper.  We are grateful to Brian Pitts and two anonymous referees for helpful comments on a previous draft.

\end{document}